\documentclass[submission,copyright,creativecommons,fleqn]{eptcs}
\providecommand{\event}{AFL 2023} 

\usepackage{amssymb,amsmath,latexsym}
\usepackage[utf8]{inputenc}
\usepackage{tikz}
\usetikzlibrary{matrix,arrows,automata,positioning}
\usepackage{stmaryrd}
\usepackage{bbding}
\usepackage{xcolor}
\usepackage{mathrsfs}
\usepackage{graphicx}

\newcommand{\beh}[1]{\llbracket #1\rrbracket}

\newtheorem{theorem}{Theorem}[section]

\newtheorem{example}[theorem]{Example}

\def\endproof{\qed\endtrivlist}
\expandafter\let\csname endproof*\endcsname=\endproof

\def\qedsymbol{\ifmmode\bgroup\else$\bgroup\aftergroup$\fi
  \vcenter{\hrule\hbox{\vrule height.6em\kern.6em\vrule}\hrule}\egroup}
\def\qed{\ifmmode\else\unskip\nobreak\fi\quad\qedsymbol}

\usepackage{xspace}

\usepackage{iftex}

\ifpdf
  \usepackage{underscore}         
  \usepackage[T1]{fontenc}        
\else
  \usepackage{breakurl}           
\fi

\DeclareFontFamily{OT1}{pzc}{}
\DeclareFontShape{OT1}{pzc}{m}{it}{<-> s * [1.200] pzcmi7t}{}
\DeclareMathAlphabet{\mathpzc}{OT1}{pzc}{m}{it}

\title{Approximate State Reduction of Fuzzy Finite Automata\footnote{This research was supported by the Science Fund of the Republic of Serbia, Grant no 7750185, Quantitative Automata Models: Fundamental Problems and Applications - QUAM}}
\author{Miroslav \'Ciri\'c\qquad Ivana Mici\'c\qquad Stefan Stanimirovi\'c
\institute{University of Ni\v s, Faculty of Sciences and Mathematics, Vi\v segradska 33, Ni\v s, Serbia}
\email{miroslav.ciric@pmf.edu.rs, ivana.micic@pmf.edu.rs, stefan.stanimirovic@pmf.edu.rs}
\and
Linh Anh Nguyen
\institute{Institute of Informatics, University of Warsaw, Banacha 2, 02-097 Warsaw, Poland, and \\
Faculty of Information Technology, Nguyen Tat Thanh University, Ho Chi Minh City, Viet Nam}
\email{nguyen@mimuw.edu.pl}
}
\def\titlerunning{Approximate state reduction of fuzzy finite automata}
\def\authorrunning{M. \'Ciri\'c, I. Mici\'c, S. Stanimirovi\'c \& L.\,A. Nguyen}
\begin{document}

\maketitle

\begin{abstract}
In this paper we introduce a new type of approximate state reductions where the behaviors of the reduced and the original automaton do not have to be identical, but they must match on all words of length less than or equal to some given natural number.~We provide four methods for performing such reductions.
\end{abstract}

\section{Introduction}

Minimization and state reduction are related problems that belong to the fundamental problems of automata theory and have many significant applications.~Minimization is the problem of finding an automaton with a minimal number of states equivalent to a given automaton.~However,
this problem cannot be solved efficiently (in polynomial time) for fuzzy finite or nondeterministic finite automata as their particular type (cf.~\cite{LQ.15} for more details).~Therefore, the so-called state reduction problem for fuzzy finite automata is studied instead, where the goal is to find an automaton equivalent to the given automaton that is not necessarily minimal, but that can be treated sufficiently small or close enough to the minimal one when comparing the number of states. In turn, the state reduction algorithm can be performed efficiently.~\'{C}iri\'{c} et al. discussed this problem in~\cite{CSIP.07}, and then in~\cite{CSIP.10},~\cite{SCI.14} and~\cite{SCB.18}, where they proposed state reduction methods that construct sequences of fuzzy matrices. The drawback of these methods is that these sequences can be infinite when the underlying structure of truth values is not locally finite. However, even when the sequences of matrices are finite, the number of different elements in sequences can be high.~For the reasons above, different authors have proposed an approximate approach not only in the context of state reduction but also in some other close contexts, such as containment and equivalence of fuzzy automata, as well as simulations and bisimulations between fuzzy automata (see \cite{BK.09,Li.08,MJS.22,MNS.22,Nguyen.23a,NMS.23a,NMS.23b,QZP.23,SM.22,SMC.22,YL.20} and articles cited there).

In the approximate state reduction problem, the main goal is to construct an automaton with a smaller number of states than a given automaton with behaviour that does not have to be identical to the behavior of a given automaton, but ``close enough'' to it.~Dominantly, scholars have defined the close\-ness between the behaviors of two automata by the concept of the degree~of equality of fuzzy sets~\cite{BK.09,YL.20,SMC.22,Nguyen.23a}. However, when we take membership values from the real unit interval, the conventional metric on that interval can also be used to define closeness \cite{Li.08,YL.18}.~Furthermore, a different fuzzy similarity measure has been proposed recently in \cite{QZP.23,QZF.22} via relational lifting.~In this paper, we approach approximate state reductions differently.~Namely, we require that the behaviors of a given automaton and its reduced one may not be strictly equivalent, but equivalent for all words with length not exceeding a given natural number $k$.~This relaxation from the strict equivalence comes naturally, as when working with fuzzy automata in practical situations, one encounters words with finite (bounded) length. It is important to emphasize that Nguyen et al. recently employed a similar idea in~\cite{NMS.23b} to generalize fuzzy simulations and fuzzy bisimulations for fuzzy automata~\cite{Nguyen.23a}.~In this paper, $k$-equivalence means the equivalence of a given fuzzy automaton and its reduced one for all words not exceeding length $k$.~Similarly, $k$-reduction means a state reduction resulting in an automaton $k$-equivalent to the given automaton.

This paper provides four $k$-reduction methods based on state reduction methods developed in \cite{SCI.14} that output a fuzzy automaton strictly equivalent to the given fuzzy automaton. Precisely, the first two methods consist of constructing a descending sequence of fuzzy quasi-order matrices. Theorems \ref{th:ri.k-e} and \ref{th:li.k-e} prove that~the fuzzy automaton formed by the different row vectors of the $k$th member of that sequence of matrices (the numbering starts from $0$) is $k$-equivalent to the given fuzzy automaton. Moreover, if the number of different elements in this sequence is not greater than $k$, then the resulting fuzzy automaton is also strictly equivalent to the given fuzzy automaton. In locally finite structures, such as the G\"odel or {\L}ukasiewicz structure, the sequence necessarily has a finite number of different elements. Therefore, we can always pick a sufficiently high $k \in \mathbb{N}$ in these structures, so the resulting fuzzy automaton is also strictly equivalent to the given fuzzy automaton. On the other hand, for some non-locally structures satisfying some additional conditions~\cite{CSIP.07,SCI.14}, such as the product structure, a reduced fuzzy automaton strictly equivalent to the original fuzzy automaton can be constructed from different row vectors of the infimum of this sequence. Therefore, by choosing the $k$th member of the sequence, the resulting $k$-equivalent reduced fuzzy automaton can be regarded as an approximation of the reduced fuzzy automaton strictly equivalent to the original fuzzy automaton.

The other two methods for performing state reduction introduced in \cite{SCI.14} consist of constructing a family of fuzzy quasi-order matrices such that a reduced fuzzy automaton built from different row vectors of the infimum of this family is strictly equivalent to the original fuzzy automaton.~These methods generally give better reductions than the first two, but their time complexity is generally higher.~Here we transform that family into a sequence of fuzzy quasi-order matrices and prove that the fuzzy automaton formed by different row vectors of the $k$th member of that sequence (the numbering again starts from $0$) is $k$-equivalent to the original automaton (Theorems \ref{th:wri.k-e} and \ref{th:wli.k-e}).~We also point out the advantages and disadvantages of the four proposed methods.

\section{Preliminaries}

Throughout this paper, $\mathbb{N}$ will denote the set of all natural numbers (including the zero).

A \textit{resuduated lattice} is defined as an algebra $\mathbb{L}=(\mathbb{L},\lor,\land,\otimes,\to,0,1)$, with four binary operations~and two constants $0$ and $1$, which satisfies the following conditions:
\begin{itemize}\parskip=-2pt
    \item[(R1)] $(\mathbb{L},\lor,\land,0,1)$ is a bounded lattice with the least element $0$ and the greatest element $1$;
    \item[(R2)] $(\mathbb{L},\otimes ,1)$ is a commutative semigroup with the identity $1$;
    \item[(R3)] the pair $(\otimes ,\to)$ satisfies the \textit{adjunction} or \textit{residuation property}: for all $a,b,c\in \mathbb{L}$, 
    \[
    a\otimes b\leqslant c \ \ \Leftrightarrow\ \ a\leqslant b\to c .
    \]
\end{itemize}
Here $\leqslant $ stands for the ordering in the lattice from (R1). The operation $\to $ is called the~\textit{residuum}, and $\otimes $ is called the \textit{multiplication}. If the bounded lattice from (R1) is complete, then $\mathbb{L}$ is called a \textit{complete residuated lattice}. As it is customary in the theory of algebraic structures to use the same notation for an algebra and its carrier set, here we denote the residuated lattice and its carrier set with the same symbol $\mathbb{L}$ as well.

The main examples of complete residuated lattices are those whose carrier set is the real unit interval $\mathbb{I}=[0,1]$ and the multiplication is some triangular norm on $\mathbb{I}$, such as, for example, the Gödel structure, product structure and {\L}ukasiewicz structure. For more information about complete residuated lattices and the mentioned structures on $\mathbb{I}$ we refer to the books \cite{Bel.02,BV.05} and other papers listed below, in the list of references.

Let $\mathbb{L}$ be an arbitrary complete residuated lattice.~For  arbitrary $m,n\in \mathbb{N}\setminus\{0\}$, by $\mathbb{L}^{\!\!m\times n}$ we denote the set of all $m\times n$ matrices with entries in $\mathbb{L}$, and by $\mathbb{L}^{\!\!n}$ the set of all vectors of size $n$ with entries in $\mathbb{L}$ (by the size of a vector we mean the number of its coordinates).~A \textit{fuzzy subset} of a non-empty set $A$ is defined as any function $\alpha:A\to \mathbb{L}$, and a \textit{fuzzy relation} on $A$ is defined as any fuzzy subset of $A\times A$, that is, as any function $M:A\times A\to \mathbb{L}$.~For $a\in A$, the value $\alpha(a)$ is called the \textit{membership degree} of $a$ in the fuzzy set $\alpha $.~Here we deal mostly with fuzzy subsets of a finite set, as well as with fuzzy relations on a finite set, and then, when dealing with a finite set $A=\{a_1,a_2,\ldots ,a_n\}$, a fuzzy subset $\alpha $ of $A$ is identified with a vector from $\mathbb{L}^{\!\!n}$ whose $i$th coordinate is $\alpha (a_i)$, while a fuzzy relation $M$ on $A$ is identified with a matrix from $\mathbb{L}^{\!\!n\times n}$ whose $(i,j)$-entry is $M(a_i,a_j)$.~Without risk of confusion, the vector corresponding to the fuzzy subset $\alpha $ is denoted by the same symbol $\alpha$, and its $i$th coordinate is denoted by $\alpha(i)$, while the matrix corresponding to the fuzzy relation $M$ is denoted by the same symbol $M$, and its $(i,j)$-entry is denoted by $M(i,j)$.

For a matrix $M\in \mathbb{L}^{\!\!n\times n}$ and a fixed $i\in [1..n]$, where $[1..n]=\{1,2,\ldots ,n\}$, the vector whose $j$th coordinate is $M(i,j)$, for~any $j\in [1..n]$, is called the $i$th \textit{row vector} of $M$, and for a fixed $j\in [1..n]$, the vector whose $i$th coordinate is $M(i,j)$, for any $i\in [1..n]$, is called the $j$th \textit{column vector} of $M$.

The product $M\cdot N$ of two matrices $M,N\in \mathbb{L}^{\!\!n\times n}$ (fuzzy relations on $A$) is a matrix from $\mathbb{L}^{\!\!n\times n}$ (a fuzzy relation on $A$) defined by
\[
(M\cdot N) (i,j)=\bigvee_{s=1}^{\!\!n} M(i,s)\otimes N(s,j),
\]
for all $i,j\in [1..n]$, the products $\alpha\cdot M$ and $M\cdot \beta $ of vectors $\alpha,\beta\in \mathbb{L}^{\!\!n}$ (fuzzy subsets of $A$) and a matrix $M\in \mathbb{L}^{\!\!n\times n}$ (fuzzy relation on $A$) are vectors from $\mathbb{L}^{\!\!n}$ (fuzzy subsets of $A$) defined by
\[
(\alpha\cdot M) (i)=\bigvee_{s=1}^{\!\!n} \alpha(s)\otimes M(s,i), \qquad 
(M\cdot \beta) (i)=\bigvee_{s=1}^{\!\!n} M(i,s)\otimes \beta(s),
\]
for every $i\in [1..n]$, and the product $\alpha\cdot \beta $ of two vectors from $\mathbb{L}^{\!\!n}$ (fuzzy subsets of $A$) is the element from $\mathbb{L}$ defined by
\[
\alpha\cdot\beta =\bigvee_{s=1}^{\!\!n} \alpha(s)\otimes \beta (s).
\]
The last product $\alpha\cdot \beta $ is called the \textit{scalar product} or \textit{dot product} of vectors (fuzzy subsets) $\alpha $ and $\beta$. 

The ordering $\leqslant $ on $\mathbb{L}^{\!\!n\times n}$ is defined entrywise by
\[
M\leqslant N\ \ \Leftrightarrow\ \ M(i,j)\leqslant N(i,j), \ \text{for all $i,j\in [1..n]$},
\]
for all $M,N\in \mathbb{L}^{\!\!n\times n}$, and similarly, the ordering $\leqslant $ on $\mathbb{L}^{\!\!n}$ is defined coordinatewise by
\[
\alpha\leqslant \beta\ \ \Leftrightarrow\ \ \alpha(i)\leqslant \beta(i), \ \text{for each $i\in [1..n]$},
\]
for all $\alpha,\beta\in \mathbb{L}^{\!\!n}$. It is easy to verify that these
orderings on $\mathbb{L}^{\!\!n\times n}$ and $\mathbb{L}^{\!\!n}$ are compatible with matrix products and vector-matrix products, that is,
\begin{align*}
    \alpha\leqslant \beta \ \ &\Rightarrow\ \ \alpha\cdot M\leqslant \beta\cdot M,\ \ M\cdot \alpha\leqslant M\cdot \beta \\
    M\leqslant N\ \ &\Rightarrow\ \ K\cdot M\leqslant K\cdot N, \ \ M\cdot K\leqslant N\cdot K,\ \ \alpha\cdot M\leqslant \alpha\cdot N,\ \ M\cdot \alpha\leqslant N\cdot \alpha
\end{align*}
for all $\alpha,\beta \in \mathbb{L}^{\!\!n}$ and $K,M,N\in \mathbb{L}^{\!\!n\times n}$. The supremum and infimum of a family $\{M_s\}_{s\in I}$ of matrices from $\mathbb{L}^{\!\!n\times n}$ are respectively matrices from $\mathbb{L}^{\!\!n\times n}$ defined by
\[
\bigl(\bigvee_{s\in I}M_s\bigr)(i,j)=\bigvee_{s\in I} M_s(i,j), \qquad 
\bigl(\bigwedge_{s\in I}M_s\bigr)(i,j)=\bigwedge_{s\in I} M_s(i,j),
\]
for all $i,j\in [1..n]$.

The matrix $I_n\in \mathbb{L}^{\!\!n\times n}$ defined by $I_n(i,j)=1$, for $i=j$, and $I_n(i,j)=0$, for $i\ne j$, $i,j\in [1..n]$, is called the \textit{identity matrix} of order $n$. A matrix $M\in \mathbb{L}^{\!\!n\times n}$ is \textit{reflexive} if $I_n\leqslant M$, it is \textit{transitive} if $M^2\leqslant M$, where $M^2=M\cdot M$, and it is \textit{symmetric} if $M(i,j)=M(j,i)$, for all $i,j\in [1..n]$. A reflexive~and~transitive matrix is called a \textit{fuzzy quasi-order matrix}, and a symmetric fuzzy quasi-order matrix is called a \textit{fuzzy equivalence matrix}. 

For matrices $M,N\in \mathbb{L}^{\!\!n\times n}$, the \textit{right residual} of $N$ by $M$ is a matrix $M\backslash N\in \mathbb{L}^{\!\!n\times n}$ defined by
\begin{equation}\label{eq:r.res}
    (M\backslash N)(j,k)=\bigwedge_{i=1}^n M(i,j)\to N(i,k),
\end{equation}
for all $j,k\in [1..n]$, and the \textit{left residual} of $N$ by $M$ is a matrix $N/M\in \mathbb{L}^{\!\!n\times n}$ defined by
\begin{equation}\label{eq:l.res}
    (N/M)(i,j)=\bigwedge_{k=1}^n M(j,k)\to N(i,k),
\end{equation}
for all $i,j\in [1..n]$.~Matrix residuals are related with matrix multiplication by the following \textit{residuation} (\textit{adjunction}) \textit{property}:
\begin{equation}\label{eq:res.prop}
    K\cdot M \leqslant N\ \ \Leftrightarrow\ \ M\leqslant K\backslash N\ \ \Leftrightarrow\ \ K\leqslant N/M,
\end{equation}
for all $K,M,N\in \mathbb{L}^{\!\!n\times n}$.~Next, for $\alpha,\beta\in \mathbb{L}^{\!\!n}$, the 
\textit{right residual} of $\beta $ by $\alpha $ is a matrix $\alpha\backslash\beta \in \mathbb{L}^{\!\!n\times n}$ given by
\begin{equation}\label{eq:r.res.vec}
    (\alpha\backslash\beta)(i,j)=\alpha(i)\to \beta(j),
\end{equation}
for all $i,j\in [1..n]$, and the \textit{left residual} of $\beta $ by $\alpha $ is a matrix $\beta/\alpha \in \mathbb{L}^{\!\!n\times n}$ given by
\begin{equation}\label{eq:l.res.vec}
    (\beta/\alpha)(j,i)=\alpha(i)\to \beta(j),
\end{equation}
for all $i,j\in [1..n]$. It is clear that $\alpha\backslash\beta=(\beta/\alpha)^\top$, that is, $\beta/\alpha=(\alpha\backslash\beta)^\top$ (here $M^\top $ denotes the \textit{transpose}~of a matrix $M$).~The \textit{residuation property} for these residuals is
\begin{equation}\label{eq:res.prop.vec}
    \alpha \cdot M\leqslant \beta \ \ \Leftrightarrow\ \ M \leqslant \alpha\backslash \beta , \qquad \qquad M\cdot \alpha\leqslant \beta \ \ \Leftrightarrow\ \ M\leqslant \beta/\alpha ,
\end{equation}
for all $\alpha,\beta\in \mathbb{L}^{\!\!n}$ and $M\in \mathbb{L}^{\!\!n\times n}$.

The~representation of the matrix $M\in \mathbb{L}^{\!\!m\times n}$ in the form of the product $M=L\cdot R$, where $L\in \mathbb{L}^{\!\!m\times r}$ and $R\in \mathbb{L}^{\!\!r\times n}$, is called the \textit{$r$-factorization} of that matrix $M$.~The smallest number $r$ for which there is an $r$-factori\-zation of the matrix $M$ is denoted by $\varrho(M)$ and is called the \textit{Schein's rank} or \textit{factor rank}~of~$M$. The concepts of $r$-factorization and Schein's rank are defined for arbitrary matrices, but have particularly good properties when applied to fuzzy quasi-order matrices (cf. \cite{SCB.18}).~As shown in \cite{ICST.15,SCI.14},~a~fuzzy quasi-order matrix $Q$ has the same number of different row vectors and different column vectors, which is denoted by $d(Q)$.~In the general case, $\varrho(Q)\leqslant d(Q)$, but for every fuzzy quasi-order matrix $Q$ with entries in~a complete residuated lattice $\mathbb{L}$ in which for all $a,b\in \mathbb{L}$ from $a\lor b=1$ it follows $a=1$ or $b=1$ (for instance, this holds if $\mathbb{L}$ is linearly ordered), we have that $\varrho(Q)= d(Q)$ (cf. \cite{SCB.18}).

For undefined notions and notation we refer to \cite{Bel.02,BV.05}.

\section{Fuzzy finite automata and the state reduction problem}

Throughout this paper, if not noted otherwise, $\mathbb{L}$ will denote an arbitrary complete residuated lattice and $X$ will denote an arbitrary non-empty alphabet. By $X^*$ we denote the free monoid over $X$, whose identity, called the \textit{empty word}, is denoted by $\varepsilon$, while by $X^+$ we denote the free semigroup over $X$, i.e., $X^+=X^*\setminus\{\varepsilon\}$.

A \textit{fuzzy finite automaton} over $\mathbb{L}$ and $X$ is defined as a tuple $\mathpzc{A}=(A,\sigma,\delta,\tau)$, where $A$ is a non-empty finite set, while $\sigma$, $\tau$ and $\delta $ are functions such that $\sigma,\tau:A\to \mathbb{L}$ and $\delta:A\times X\times A\to \mathbb{L}$.~The function $\delta $ is often replaced by the family of functions $\{\delta_x\}_{x\in X}$, where $\delta_x:A\times A\to \mathbb{L}$ is given by $\delta_x(a,b)=\delta(a,x,b)$, for all $a,b\in A$ and $x\in X $. We call $A$ the \textit{set of states}, $\sigma $ the \textit{fuzzy set of initial states}, $\tau $ the \textit{fuzzy set of terminal states}, and $\delta $ and $\delta_x$, $x\in X$, the \textit{fuzzy transition functions}. The number of states of $\mathpzc{A}$ will be denoted by $|\mathpzc{A}|$.

The \textit{behavior} of the fuzzy finite automaton $\mathpzc{A}$ is a function $\beh{\mathpzc{A}}:X^*\to \mathbb{L}$ (i.e., a fuzzy subset of $X^*$) defined by
\begin{equation}\label{eq:behA}
\beh{\mathpzc{A}}(u)=\bigvee_{(a_0,a_1,\ldots,a_k)\in A^{k+1}} \sigma(a_0)\otimes \delta(a_0,x_1,a_1)\otimes \delta(a_1,x_2,a_2) \cdots \delta(a_{k-1},x_k,a_k)\otimes \tau (a_k), 
\end{equation}
for $u=x_1x_2\dots x_k\in X^+$, $x_1,x_2,\ldots ,x_k\in X$, and 
\begin{equation}\label{eq:behAe}
\beh{\mathpzc{A}}(\varepsilon)=\bigvee_{a\in A} \sigma(a)\otimes \tau (a). 
\end{equation}
We say that $\beh{\mathpzc{A}}$ is the \textit{fuzzy language recognized (accepted) by} the fuzzy finite automaton $\mathpzc{A}$,~or~in~short just the \textit{fuzzy language of} $\mathpzc{A}$.

As we noted in the previous section, fuzzy subsets of a finite set with $n$ elements can be treated as fuzzy vectors. i.e., as vectors from $\mathbb{L}^{\!\!n}$, while fuzzy relations on such a set can be treated as fuzzy matrices, i.e. as matrices from $\mathbb{L}^{\!\!n\times n}$.~When such a way of viewing is applied to the fuzzy finite autom\-aton~$\mathpzc{A}$, then $\sigma $ and $\tau $ are treated as vectors from $\mathbb{L}^{\!\!n}$, called respectively the \textit{initial fuzzy vector} and \textit{terminal fuzzy vector}, while    $\delta_x$, $x\in X$, are treated as matrices from $\mathbb{L}^{\!\!n\times n}$, called the \textit{transition fuzzy matrices}.~Then the tuple
 is called the \textit{linear representation} of the fuzzy finite autom\-aton $\mathpzc{A}$, while $n$ is called the \textit{dimension} of $\mathpzc{A}$.~Such a way of looking at $\sigma $, $\tau $ and $\delta_x$'s will be applied here as well. In the vector-matrix form, the behavior of the fuzzy finite automaton $\mathpzc{A}$ is represented as follows:
\begin{equation}\label{eq:behAl}
\beh{\mathpzc{A}}(u)=\sigma \cdot \delta_{x_1}\cdot \delta_{x_2}\cdots \delta_{x_s}\cdot \tau = \sigma\cdot \delta_u\cdot \tau , 
\end{equation}
for $u=x_1x_2\dots x_s\in X^+$, $x_1,x_2,\ldots ,x_s\in X$, where $\delta_u=\delta_{x_1}\cdot \delta_{x_2}\cdots \delta_{x_k}$, and 
\begin{equation}\label{eq:behAel}
\beh{\mathpzc{A}}(\varepsilon)=\sigma \cdot \tau . 
\end{equation}
As can be seen, here we treat $\sigma $ as a row vector and $\tau $ as a column vector.

Two fuzzy finite automata $\mathpzc{A}$ and $\mathpzc{B}$ over $\mathbb{L}$ and $X$ are said to be \textit{equivalent} if 
\begin{equation}\label{eq:eqiuv}
    \beh{\mathpzc{A}}(u)=\beh{\mathpzc{B}}(u),\qquad \text{for every $u\in X^*$},
\end{equation}
i.e., if $\beh{\mathpzc{A}}=\beh{\mathpzc{B}}$, and for an arbitrary $k\in \mathbb{N}$, we say that $\mathpzc{A}$ and $\mathpzc{B}$ are  \textit{$k$-equivalent}~if 
\begin{equation}\label{eq:k.eqiuv}
    \beh{\mathpzc{A}}(u)=\beh{\mathpzc{B}}(u),\qquad \text{for every $u\in X^*$ such that $|u|\leqslant k$}.
\end{equation}
If $\mathpzc{A}$ and $\mathpzc{B}$ are equivalent or $k$-equivalent, we also say that $\mathpzc{A}$ is \textit{equivalent to} or \textit{$k$-equivalent to} $\mathpzc{B}$, and vice versa.

Let $Q\in \mathbb{L}^{\!\!n\times n}$ be an arbitrary fuzzy quasi-order matrix. Let $1 \le i_1 \le \ldots \le i_k \le n$ be the smallest indices of all pairwise distinct rows of $Q$. As mentioned earlier \cite{ICST.15,SCI.14}, they are also the smallest indices of all pairwise distinct columns of $Q$. Let $Q_r \in \mathbb{L}^{\!\!k\times n}$ (respectively, $Q_c \in \mathbb{L}^{\!\!n\times k}$) be the matrix consisting of the rows (respectively, columns) $i_1, \ldots, i_k$ of $Q$. 
Let $\mathpzc{A}=(A,\sigma,\delta,\tau)$ be a fuzzy finite automaton over $\mathbb{L}$ and $X$ with $n$ states. Then, a new fuzzy finite automaton, with the linear representation
\[
\mathpzc{A}_Q=(k,\sigma^Q,\{\delta^Q_x\}_{x\in X},\tau^Q),
\]
is defined by putting that 
\begin{align*}
    &\sigma^{Q} = \sigma \cdot Q_c, \\
    &\delta^{Q}_x =  Q_r \cdot \delta_x \cdot Q_c,\quad \textrm{for all $x\in X$}, \\
    &\tau^{Q} =  Q_r \cdot \tau. 
\end{align*}
It is clear that $\sigma^Q,\tau^Q\in \mathbb{L}^{\!\!k}$ and  $\delta_x^Q\in \mathbb{L}^{\!\!k\times k}$, for each $x\in X$, so  $\mathpzc{A}_{Q}$ is a well-defined fuzzy finite automaton. Since $Q=Q_c\cdot Q_r$ (cf.~\cite[Theorem 4.1]{SCB.18}, the behavior of $\mathpzc{A}_{Q}$ is given in the vector-matrix form by
\begin{equation}\label{eq:behAQ}
    \beh{\mathpzc{A}_{Q}}(u)=\sigma\cdot Q\cdot \delta_{x_1}\cdot Q\cdot \delta_{x_2}\cdot \cdots \cdot Q\cdot \delta_{x_s}\cdot Q\cdot \tau ,
\end{equation}
for every $u=x_1x_2\ldots x_s\in X^+$, where $x_1,x_2,\ldots ,x_s\in X$, and 
\begin{equation}\label{eq:behAQe}
    \beh{\mathpzc{A}_{Q}}(\varepsilon)=\sigma\cdot Q\cdot \tau .
\end{equation}
As we mentioned earlier, fuzzy matrices can be identified with fuzzy relations between finite sets. Since in the theory of fuzzy sets, the "rows" of fuzzy relations are known as \textit{aftersets}, and the "columns" are known as \textit{foresets}, the fuzzy finite automaton $\mathpzc{A}_{Q}$ was called in \cite{SCI.14} the \textit{afterset automaton} of $\mathpzc{A}$ corres\-ponding to the fuzzy quasi-order matrix $Q$, but it can also be rightly called the \textit{row automaton}.~It~was also shown in \cite{SCI.14} that if in the construction of the automaton $\mathpzc{A}_{Q}$ instead of the rows (aftersets) of the matrix Q we use its columns (foresets), then essentially nothing changes, because an isomorphic fuzzy finite automaton is obtained. 

Let us note that the number of states of $\mathpzc{A}_{Q}$ is $d(Q)\leqslant n$, that is, $\mathpzc{A}_{Q}$ has less than or equal number of states with $\mathpzc{A}$. Therefore, by constructing the automaton $\mathpzc{A}_{Q}$ we can reduce the number of states of the automaton $\mathpzc{A}$, provided that these two automata are equivalent.~Consequently, the main question here is \textit{under what conditions on $Q$ the automaton $\mathpzc{A}_{Q}$ is equivalent to $\mathpzc{A}$}?~This question can also~be formulated as: \textit{under what conditions on $Q$ the construction of $\mathpzc{A}_{Q}$ preserves the fuzzy language of the automaton $\mathpzc{A}$}?~More answers to these~questions were provided in \cite{CSIP.10} and \cite{SCI.14}.~\textit{Right invariant matrices} were defined in \cite{CSIP.10,SCI.14} as solutions of the system of matrix equations
\begin{align*}
    &(\text{ri-1})\ \ U\cdot \tau \leqslant \tau; \\
    &(\text{ri-2})\ \ U\cdot \delta_x\leqslant \delta_x\cdot U, \quad \text{for all $x\in X$},
\end{align*}
where $U$ is an unknown matrix taking values in $\mathbb{L}^{\!\!n\times n}$, and \textit{left invariant matrices} were defined as solutions of the system
\begin{align*}
    &(\text{li-1})\ \ \sigma\cdot U\leqslant \sigma; \\
    &(\text{li-2})\ \ \delta_x\cdot U\leqslant U\cdot \delta_x, \quad \text{for all $x\in X$}.
\end{align*}
It was proved in \cite{CSIP.10} that both right and left invariant fuzzy equivalence matrices provide equivalence between $\mathpzc{A}$ and the corresponding row automaton, while the same for fuzzy quasi-order matrices was proved in \cite{SCI.14}.~At the same time, it has been proven that fuzzy quasi-order matrices give better reductions than fuzzy equivalence matrices, in the sense that they produce fuzzy finite automata with a smaller number of states.

Procedures for computing the greatest right and left invariant matrices, which are necessarily fuzzy quasi-order matrices, were provided in \cite{SCI.14}. For right invariant matrices, the procedure consists of building a decreasing sequence of matrices, which is defined in the next section by formula \eqref{eq:seq.Q}. When there are two equal consecutive members of that sequence, the sequence is finite and stabilizes at some member which is the greatest right invariant matrix.~For instance, if $\mathbb{L}$ is the 
G\"odel structure or Boolean algebra, then every such sequence is finite and the greatest right invariant matrix can be computed in a finite number of steps.~However, there are also cases when this sequence is infinite and the greatest right invariant matrix can not be computed in a finite number of steps.~For example, this may happen when $\mathbb{L}$ is the 
product structure.~All this also applies to left-invariant matrices, where the procedure for computing the greatest such matrix is based on the decreasing sequence of matrices defined in the next section by formula \eqref{eq:seq.P}.

In paper \cite{SCI.14}, another way was also provided to get a fuzzy quasi-order matrix $Q$ for which the~construction of $\mathpzc{A}_{Q}$ will preserve the language of $\mathpzc{A}$.~\textit{Weakly right invariant matrices} were defined in as solutions of the system of matrix equations
\begin{align*}
    &(\text{wri})\ \ U\cdot \tau_u \leqslant \tau_u, \quad \text{for all $u\in X^*$},
\end{align*}
where $U$ is an unknown matrix taking values in $\mathbb{L}^{\!\!n\times n}$ and $\tau_u=\delta_u\cdot \tau$, and \textit{weakly left invariant matrices} were defined as solutions of the system
\begin{align*}
    &(\text{wli})\ \ \sigma_u\cdot U\leqslant \sigma_u,\quad \text{for all $u\in X^*$},
\end{align*}
where $\sigma_u=\sigma\cdot \delta_u$.~As shown in \cite{SCI.14}, both for any weakly right invariant or weakly left invariant fuzzy quasi-order matrix $Q$, the row automaton $\mathpzc{A}_Q$ is equivalent to $\mathpzc{A}$, and the greatest weakly right invariant fuzzy quasi-order, i.e., the greatest solution of (wri), can be expressed as
\begin{equation}\label{eq:gwri}
    \bigwedge_{u\in X^*} \tau_u\slash \tau_u ,
\end{equation}
while the greatest weakly left invariant fuzzy quasi-order, i.e., the greatest solution of (wli), can be~expressed as
\begin{equation}\label{eq:gwli}
    \bigwedge_{u\in X^*} \sigma_u\backslash \sigma_u .
\end{equation}
In general, the greatest weakly right invariant fuzzy quasi-order matrix provides better reduction than the greatest right invariant one, but its computation may be significantly more complex.~There may also be a problem of efficient computation of these matrices, because the families $\{\tau_u\mid u\in X^*\}$ and $\{\sigma_u\mid u\in X^*\}$ can be infinite, and even when they are finite, the number of their~mem\-bers can be too large.

All the mentioned problems that concern the computation of the greatest invariant and weakly invariant fuzzy quasi-order matrices actualize the issue of approximate state reductions of fuzzy finite automata, which will be discussed in the next section.

Note that in \cite{SCB.18} a method for additional reduction of the number of states of the automaton $\mathpzc{A}_Q$ is offered, which is based on the $r$-factorization of the fuzzy quasi-order matrix $Q$. Namely, let $Q=L\cdot R$ be an $r$-factorization of $Q$, i.e., $L\in \mathbb{L}^{n\times r}$ and $R\in \mathbb{L}^{r\times n}$. Then we can construct a fuzzy finite automaton $\mathpzc{A}_{L|R}$ with the linear representation 
\[
\mathpzc{A}_{L|R}=(r,\sigma^{L|R},\{\delta_x^{L|R}\}_{x\in X},\tau^{L|R}),
\]
where 
\begin{align*}
&\sigma^{L|R}=\sigma\cdot L,\\
&\delta_x^{L|R}=R\cdot \delta_x\cdot L,\quad \text{for all $x\in X$}, \\
&\tau^{L|R}=R\cdot \tau .
\end{align*}
Then $\sigma^{L|R},\tau^{L|R}\in \mathbb{L}^r$ and $\delta_x^{L|R}\in \mathbb{L}^{r\times r}$, so $\mathpzc{A}_{L|R}$ is a well-defined fuzzy finite automaton with $r$ states. According to 
\eqref{eq:behAQ}, \eqref{eq:behAQe} and $Q=L\cdot R$ we have that $\beh{\mathpzc{A}_{L|R}}=\beh{\mathpzc{A}_Q}$.~Therefore, when we reduce the number of states of the fuzzy finite automaton $\mathpzc{A}$ using the greatest fuzzy quasi-order matrix $Q$, an addi\-tional reduction of the number of states could be performed with the help of an $r$-factorization of $Q$. Clearly, the smallest number of states we can obtain in this way is $\varrho(Q)$, the Schein's rank of $Q$.

\section{Approximate state reduction: $\boldsymbol{k}$-reduction}

As we already noted in the introduction, there were already several articles dealing with approximate state reduction, mainly in the context of studying approximate simulations and bisimulations between fuzzy  finite automata.~The approach used in those articles was to, starting from a given fuzzy finite automaton, construct a new fuzzy finite automaton with a smaller number of states, whose fuzzy language does not have to be identical to the fuzzy language of the original automaton, but is sufficiently similar to that fuzzy language, in relation to a certain measure of similarity.~In most papers, that measure was based on subsethood and equality degrees between fuzzy sets.

In this paper, we use a different approach.~We require that between fuzzy languages of the original and the reduced fuzzy finite automaton there is an exact match for all words of length less than or equal to $k$, while membership degrees for longer words do not matter to us. In the terminology from the previous section, this means that these fuzzy languages are $k$-equivalent.~Such an approach is quite natural if we keep in mind that in practical applications of automata the length of input words is always limited, so the only thing that matters to us is that the membership degrees match for words whose length does not exceed that limit.

A procedure for reduction of the number of states of a fuzzy finite automaton $\mathpzc{A}$, which results in a fuzzy finite automaton that is $k$-equivalent to $\mathpzc{A}$, will be called a \textit{$k$-reduction}.~The following theorem provides such a procedure that is based on the procedure for reduction of the number of states, provided in \cite{SCI.14}, that results in an automaton which is strictly equivalent to the original automaton.

\pagebreak

\begin{theorem}\label{th:ri.k-e}
Let $\mathpzc{A}=(A,\sigma,\delta,\tau)$ be a fuzzy finite automaton over $\mathbb{L}$ and $X$ with $n$ states. 

Let us inductively define a sequence of matrices $\{Q_k\}_{k\in \mathbb{N}}\subset \mathbb{L}^{\!\!n\times n}$ as follows:
\begin{equation}\label{eq:seq.Q}
 Q_0=\tau\slash \tau , \qquad Q_{k+1}=Q_k\land \bigwedge_{x\in X} [(\delta_x\cdot Q_k)\slash \delta_x], \quad\text{for every}\ k\in \mathbb{N}.
\end{equation}
Then for an arbitrary $k\in \mathbb{N}$ the following statements hold:
\begin{itemize}\parskip=-2pt
\item[{\rm (a)}] $Q_k$ is a fuzzy quasi-order matrix;
\item[{\rm (b)}] $\mathpzc{A}_{Q_k}$ is $k$-equivalent to $\mathpzc{A}$;
\item[{\rm (c)}] if $Q_k=L\cdot R$ is an $r$-factorization of $Q_k$, for some $r\leqslant d(Q_k)$, then $\mathpzc{A}_{L|R}$ is $k$-equivalent to $\mathpzc{A}$;
\item[{\rm (d)}] if $Q_s=Q_{s+1}$, for some $s\leqslant k$, then $Q_k=Q_s$ and both $\mathpzc{A}_{Q_k}$ and $\mathpzc{A}_{L|R}$ are equivalent to $\mathpzc{A}$.
\end{itemize}
\end{theorem}

\begin{proof} (a) By its definition, $Q_0$ is the greatest solution of the inequation $U\cdot \tau\leqslant \tau$, where $U$ is an unknown matrix taking values in $\mathbb{L}^{\!\!n\times n}$.~Since $I_n$ is also a solution to this inequation, we conclude that $I_n\leqslant Q_0$.~Moreover, we have that
\[
Q_0^2\cdot \tau =Q_0\cdot Q_0\cdot \tau \leqslant Q_0\cdot \tau \leqslant \tau,
\]
which means that $Q_0^2$ is also a solution of $U\cdot \tau\leqslant \tau $, and since $Q_0$ is the greatest solution to this inequation, we conclude that $Q_0^2\leqslant Q_0$. This proves that $Q_0$ is a fuzzy quasi-order matrix.

Suppose that $Q_s$ is a fuzzy quasi-order matrix, for some $s\in \mathbb{N}_0$.~Let us observe that for every $x\in X$ the matrix $M_x=(\delta_x\cdot Q_s)\slash \delta_x$ is the greatest solution of the inequation $U\cdot \delta_x\leqslant \delta_x\cdot Q_s$, with an unknown matrix $U$. We also have that $I_n\leqslant Q_s$, whence
\[
I_n\cdot \delta_x=\delta_x\cdot I_n\leqslant \delta_x\cdot Q_s,
\]
which means that $I_n$ is also a solution to the inequation $U\cdot \delta_x\leqslant \delta_x\cdot Q_s$, and consequently, $I_n\leqslant M_x$, since $M_x$ is the greatest solution of this inequation.

On the other hand, we have that
\[
M_x^2\cdot \delta_x=M_x\cdot M_x\cdot \delta_x\leqslant M_x\cdot \delta_x\cdot Q_s\leqslant \delta_x\cdot Q_s\cdot Q_s=\delta_x\cdot Q_s^2= \delta_x\cdot Q_s,
\]
which means that $M_x^2$ is also a solution of $U\cdot \delta_x\leqslant \delta_x\cdot Q_s$, and thus, $M_x^2\leqslant M_x$. Hence, $M_x=(\delta_x\cdot Q_s)\slash \delta_x$ is a fuzzy quasi-order matrix, for each $x\in X$, and the matrix $Q_{s+1}$ is also a fuzzy quasi-order~matrix, as the intersection of the family of fuzzy quasi-order matrices $M_x=(\delta_x\cdot Q_s)\slash \delta_x$, $x\in X$,~and~the fuzzy quasi-order matrix $Q_s$.

Let us also note that $Q_{s+1}$ is the greatest solution of the system of linear matrix inequations
\begin{equation}\label{eq:syst.lmi}
   U\cdot \delta_x\leqslant \delta_x\cdot Q_s,\quad x\in X , 
\end{equation}
contained in $Q_s$, where $U$ is an unknown matrix taking values in $\mathbb{L}^{\!\!n\times n}$. 

(b) According to \eqref{eq:seq.Q}, we have that $Q_0\cdot \tau \leqslant \tau $ and $Q_{t+1}\cdot \delta_x\leqslant \delta_x\cdot Q_t$, for arbitrary $t\in \mathbb{N}_0$ and $x\in X$. On the other hand, from $Q_t\leqslant Q_0$ it follows that $Q_t\cdot \tau\leqslant Q_0\cdot \tau \leqslant \tau $, whereas from $I_n\leqslant Q_t$ we obtain that $\tau=I_n\cdot \tau \leqslant Q_t\cdot \tau $.~Furthermore, since $Q_{t+1}\cdot \delta_x\leqslant \delta_x\cdot Q_t$, $Q_t^2=Q_t$ and $I_n\leqslant Q_{t+1}$, we have that
\[
Q_{t+1}\cdot \delta_x\cdot Q_t\leqslant \delta_x\cdot Q_t^2=\delta_x\cdot Q_t, \qquad
\delta_x\cdot Q_t=I_n\cdot \delta_x\cdot Q_t\leqslant Q_{t+1}\cdot \delta_x\cdot Q_t.
\]
Therefore, we have proved that
\begin{equation}\label{eq:Qs}
Q_t\cdot \tau=\tau, \quad Q_{t+1}\cdot \delta_x\cdot Q_t=\delta_x\cdot Q_t,
\end{equation}
for all $t\in \mathbb{N}_0$ and $x\in X$. 

Now we have that
\[
\beh{\mathpzc{A}_{Q_k}}(\varepsilon)=\sigma\cdot Q_k\cdot \tau = \sigma\cdot\tau = \beh{\mathpzc{A}}(\varepsilon),
\]
and for $u=x_1x_2\ldots x_s\in X^+$, where $x_1,x_2,\ldots ,x_s\in X$ and $s\leqslant k$, we have that
\[
\beh{\mathpzc{A}}(u)=\sigma\cdot \delta_{x_1}\cdot \delta_{x_2}\cdot \ldots \cdot \delta_{x_s}\cdot \tau \leqslant \sigma\cdot Q_k\cdot \delta_{x_1}\cdot Q_k\cdot \delta_{x_2}\cdot \ldots \cdot Q_k\cdot \delta_{x_s}\cdot Q_k\cdot \tau = \beh{\mathpzc{A}_{Q_k}}(u),
\]
and
\begin{align*}
    \beh{\mathpzc{A}_{Q_k}}(u)&= \sigma\cdot Q_k\cdot \delta_{x_1}\cdot Q_k\cdot \delta_{x_2}\cdot \ldots \cdot Q_k\cdot \delta_{x_s}\cdot Q_k\cdot \tau \leqslant 
    \sigma\cdot Q_s\cdot \delta_{x_1}\cdot Q_{s-1}\cdot \delta_{x_2}\cdot \ldots \cdot Q_1\cdot \delta_{x_s}\cdot Q_0\cdot \tau \\
    &= \sigma\cdot \delta_{x_1}\cdot Q_{s-1}\cdot \delta_{x_2}\cdot \ldots \cdot Q_1\cdot \delta_{x_s}\cdot Q_0\cdot \tau =\sigma\cdot \delta_{x_1}\cdot \delta_{x_2}\cdot \ldots \cdot Q_1\cdot \delta_{x_s}\cdot Q_0\cdot \tau \\
    &= \ldots = \sigma\cdot \delta_{x_1}\cdot \delta_{x_2}\cdot \ldots \cdot \delta_{x_s}\cdot Q_0\cdot \tau = \sigma\cdot \delta_{x_1}\cdot \delta_{x_2}\cdot \ldots \cdot \delta_{x_s}\cdot\tau = \beh{\mathpzc{A}}(u).
\end{align*}
Therefore, we have proved that $\beh{\mathpzc{A}_{Q_k}}(u)=\beh{\mathpzc{A}}(u)$, for every $u\in X^*$ such that $|u|=s\leqslant k$, which means that the fuzzy automaton $\mathpzc{A}_{Q_k}$ is $k$-equivalent to $\mathpzc{A}$.

(c) This follows by the fact that $\beh{\mathpzc{A}_{Q_k}}=\beh{\mathpzc{A}_{L|R}}$ (cf. \cite{SCB.18}).

(d) Assume that $Q_s=Q_{s+1}$, for some $s\leqslant k$.~As shown in \cite{SCI.14}, then $Q_s=Q_t$, for every $t\geqslant s$, $t\in \mathbb{N}_0$, and $\beh{\mathpzc{A}_{Q_s}}=\beh{\mathpzc{A}}$ holds. Therefore, we have that $Q_k=Q_s$, so $\beh{\mathpzc{A}_{Q_k}}=\beh{\mathpzc{A}_{L|R}}=\beh{\mathpzc{A}_{Q_s}}=\beh{\mathpzc{A}}$.
\end{proof}

It is worth noting that in assertion (c) of Theorem \ref{th:ri.k-e} we are talking about the $r$-decomposition of the matrix $Q_k$ for some $r\in \mathbb{N}$ for which $\varrho(Q_k)\leqslant r\leqslant d (Q_k)$, rather than on the~$\varrho(k)$-decom\-position, which would give a smallest row automaton.~The reason is that the $r$-decomposition could be done using an algorithm provided in \cite{SCB.18}, which consists in removing those row vectors of $Q_k$ that can be represented as linear combinations of other row vectors.~The result of that procedure strongly~depends on the choice and order of the vectors we remove, so that one can get an $r$-decomposition for any $r$ such that $\varrho(Q_k)\leqslant r\leqslant d (Q_k)$. 

Let us also note that the use of the $r$-decompositions in $k$-reductions (and reductions in general) is useless if the underlying complete residuated lattice $\mathbb{L}$ satisfies the condition that $a\lor b=1$ implies that $a=1$ or $b=1$, since in this case we have that $\varrho(Q_k)=d(Q_k)$ (as shown in \cite{SCB.18}), so there is no reduction whatsoever. For instance, this holds when $\mathbb{L}$ is linearly ordered, what includes all the cases when $\mathbb{L}$ is the structure defined on the real unit interval by means of triangular norms.

Similarly to Theorem \ref{th:ri.k-e}, we can prove the following theorem, which provides an alternative way for $k$-reduction.

\begin{theorem}\label{th:li.k-e}
Let $\mathpzc{A}=(A,\sigma,\delta,\tau)$ be a fuzzy finite automaton over $\mathbb{L}$ and $X$ with $n$ states. 

Let us inductively define a sequence of matrices $\{P_k\}_{k\in \mathbb{N}}\subset \mathbb{L}^{\!\!n\times n}$ as follows:
\begin{equation}\label{eq:seq.P}
 P_0=\sigma\backslash \sigma , \qquad P_{k+1}=P_k\land \bigwedge_{x\in X} [\delta_x\backslash (P_k\cdot \delta_x)], \quad\text{for every}\ k\in \mathbb{N}.
\end{equation}
Then for an arbitrary $k\in \mathbb{N}$ the following statements hold:
\begin{itemize}\parskip=-2pt
\item[{\rm (a)}] $P_k$ is a fuzzy quasi-order matrix;
\item[{\rm (b)}] $\mathpzc{A}_{P_k}$ is $k$-equivalent to $\mathpzc{A}$;
\item[{\rm (c)}] if $P_k=L\cdot R$ is an $r$-factorization of $P_k$, for some $r\leqslant d(P_k)$, then $\mathpzc{A}_{L|R}$ is $k$-equivalent to $\mathpzc{A}$;
\item[{\rm (d)}] if $P_s=P_{s+1}$, for some $s\leqslant k$, then $P_k=P_s$ and both $\mathpzc{A}_{P_k}$ and $\mathpzc{A}_{L|R}$ are equivalent to $\mathpzc{A}$.
\end{itemize}
\end{theorem}

\pagebreak

The third way to perform the $k$-reduction is given by the next theorem.

\begin{theorem}\label{th:wri.k-e}
Let $\mathpzc{A}=(A,\sigma,\delta,\tau)$ be a fuzzy finite automaton over $\mathbb{L}$ and $X$ with $n$ states. 

Let us define a sequence of matrices $\{\widehat{Q}_k\}_{k\in \mathbb{N}}\subset \mathbb{L}^{\!\!n\times n}$ as follows:
\begin{equation}\label{eq:seq.hatQ}
 \widehat{Q}_k= \bigwedge_{|u|\leqslant k} \tau_u\slash \tau_u ,
\end{equation}
for each $k\in \mathbb{N}$. Then for an arbitrary $k\in \mathbb{N}_0$ the following statements hold:
\begin{itemize}\parskip=-2pt
\item[{\rm (a)}] $\widehat{Q}_k$ is a fuzzy quasi-order matrix and $Q_k\leqslant\widehat{Q}_k$, where $Q_k$ is as in Theorem~\ref{th:ri.k-e};
\item[{\rm (b)}] $\mathpzc{A}_{\widehat{Q}_k}$ is $k$-equivalent to $\mathpzc{A}$;
\item[{\rm (c)}] if $\widehat{Q}_k=L\cdot R$ is an $r$-factorization of $\widehat{Q}_k$, for some $r\leqslant d(\widehat{Q}_k)$, then $\mathpzc{A}_{L|R}$ is $k$-equivalent to $\mathpzc{A}$.
\end{itemize}
\end{theorem}

\begin{proof}
    (a) We have that $\widehat{Q}_k$ is a fuzzy quasi-order matrix as the infimum of a family of fuzzy quasi-order matrices $\tau_u\slash \tau_u$, $u\in X^*$, $|u|\leqslant k$.

    Consider an arbitrary $u\in X^*$ such that $|u|\leqslant k$.~Then $u=x_1x_2\ldots x_s$, where $x_1,x_2,\ldots ,x_s\in X$ and $s\leqslant k$, and we get
    \begin{align*}
        Q_k\cdot \tau_u&=Q_k\cdot \delta_{x_1}\cdot \delta_{x_2}\cdot \ldots \cdot \delta_{x_s}\cdot \tau \leqslant 
        Q_s\cdot \delta_{x_1}\cdot \delta_{x_2}\cdot \ldots \cdot \delta_{x_s}\cdot \tau \leqslant 
        \delta_{x_1}\cdot Q_{s-1}\cdot \delta_{x_2}\cdot \ldots \cdot \delta_{x_s}\cdot \tau \\
        &\leqslant \delta_{x_1}\cdot \delta_{x_2}\cdot Q_{s-2}\cdot \ldots \cdot \delta_{x_s}\cdot \tau \leqslant \ldots \leqslant 
        \delta_{x_1}\cdot \delta_{x_2}\cdot \ldots \cdot Q_{1}\cdot \delta_{x_s}\cdot \tau \leqslant 
        \delta_{x_1}\cdot \delta_{x_2}\cdot \ldots \cdot \delta_{x_s}\cdot Q_{0}\cdot \tau \\
        &\leqslant \delta_{x_1}\cdot \delta_{x_2}\cdot \ldots \cdot \delta_{x_s}\cdot \tau = \tau_u .
    \end{align*}
    Therefore, $Q_k\cdot \tau_u\leqslant \tau_u$. i.e., $Q_k\leqslant \tau_u\slash \tau_u$, for every $u\in X^*$ such that $|u|\leqslant k$, so
    \[
    Q_k\leqslant \bigwedge_{|u|\leqslant k} \tau_u\slash \tau_u = \widehat{Q}_k .
    \]

(b) According to the definition of $\widehat{Q}_k$ we obtain that $\widehat{Q}_k\cdot \tau_u\leqslant \tau_u$, for every $u\in X^*$, $|u|\leqslant k$, and since $\tau_u=I_n\cdot \tau_u \leqslant \widehat{Q}_k\cdot \tau_u$, we conclude that 
\begin{equation}\label{eq:eq.tau}
    \widehat{Q}_k\cdot \tau_u=\tau_u,\quad\text{for every $u\in X^*$, $|u|\leqslant k$}.
\end{equation}

Consider again an arbitrary $u\in X^*$ such that $|u|\leqslant k$. According to \eqref{eq:eq.tau} we get
\begin{align*}
    \beh{\mathpzc{A}_{\widehat{Q}_k}}(u)&= \sigma\cdot\widehat{Q}_k\cdot \delta_{x_1}\cdot\widehat{Q}_k\cdot \delta_{x_2}\cdot\widehat{Q}_k\cdot \ldots \cdot\widehat{Q}_k\cdot \delta_{x_s}\cdot\widehat{Q}_k\cdot \tau = 
    \sigma\cdot\widehat{Q}_k\cdot \delta_{x_1}\cdot\widehat{Q}_k\cdot \delta_{x_2}\cdot\widehat{Q}_k\cdot \ldots \cdot\widehat{Q}_k\cdot \delta_{x_s}\cdot\tau \\
    &= \sigma\cdot\widehat{Q}_k\cdot \delta_{x_1}\cdot\widehat{Q}_k\cdot \delta_{x_2}\cdot\widehat{Q}_k\cdot \ldots \cdot\widehat{Q}_k\cdot \tau_{x_s} = 
    \sigma\cdot\widehat{Q}_k\cdot \delta_{x_1}\cdot\widehat{Q}_k\cdot \delta_{x_2}\cdot\widehat{Q}_k\cdot \ldots \cdot \tau_{x_s}= \ldots \\
    &= \sigma\cdot\widehat{Q}_k\cdot \delta_{x_1}\cdot\widehat{Q}_k\cdot \delta_{x_2}\cdot\widehat{Q}_k\cdot \tau_{x_3\ldots x_s} =
    \sigma\cdot\widehat{Q}_k\cdot \delta_{x_1}\cdot\widehat{Q}_k\cdot \delta_{x_2}\cdot \tau_{x_3\ldots x_s} = \\
    &=\sigma\cdot\widehat{Q}_k\cdot \delta_{x_1}\cdot\widehat{Q}_k \cdot \tau_{x_2x_3\ldots x_s} = \sigma\cdot\widehat{Q}_k\cdot \delta_{x_1}\cdot \tau_{x_2x_3\ldots x_s} = \sigma\cdot\widehat{Q}_k\cdot  \tau_{x_1x_2x_3\ldots x_s} \\
    &= \sigma\cdot  \tau_{x_1x_2x_3\ldots x_s} = \sigma\cdot \tau_u=\beh{\mathpzc{A}}(u).
\end{align*}
Therefore, we have proved that $\mathpzc{A}_{\widehat{Q}_k}$ is $k$-equivalent to $\mathpzc{A}$.

(c) This is proved in the same way as the statement (c) in Theorem \ref{th:ri.k-e}.
\end{proof}

In a similar way we can prove the theorem that provides the fourth way to perform the $k$-reduction.

\begin{theorem}\label{th:wli.k-e}
Let $\mathpzc{A}=(A,\sigma,\delta,\tau)$ be a fuzzy finite automaton over $\mathbb{L}$ and $X$ with $n$ states. 

Let us define a sequence of matrices $\{\widehat{P}_k\}_{k\in \mathbb{N}}\subset \mathbb{L}^{\!\!n\times n}$ as follows:
\begin{equation}\label{eq:seq.hatP}
 \widehat{P}_k= \bigwedge_{|u|\leqslant k} \sigma_u\backslash \sigma_u ,
\end{equation}
for every $k\in \mathbb{N}$.

\pagebreak

\noindent
Then for an arbitrary $k\in \mathbb{N}$ the following statements hold:
\begin{itemize}\parskip=-1pt
\item[{\rm (a)}] $\widehat{P}_k$ is a fuzzy quasi-order matrix and $P_k\leqslant\widehat{P}_k$, where $P_k$ is as in Theorem~\ref{th:li.k-e};
\item[{\rm (b)}] $\mathpzc{A}_{\widehat{P}_k}$ is $k$-equivalent to $\mathpzc{A}$;
\item[{\rm (c)}] if $\widehat{P}_k=L\cdot R$ is an $r$-factorization of $\widehat{P}_k$, for some $r\leqslant d(\widehat{P}_k)$, then $\mathpzc{A}_{L|R}$ is $k$-equivalent to $\mathpzc{A}$.
\end{itemize}
\end{theorem}

The previous four theorems provide four different methods for $k$-reduction of fuzzy finite autom\-ata. The question naturally arises: Do we really need all four methods?~In other words, are any~of~these methods better than others, so that others are not necessary?~In the sequel, we will show that each of these methods has some advantages, but also some disadvantages in relation to the others, as~well~as in relation to methods for full state reduction (reduction that ensures full equivalence) provided~in~\cite{SCI.14}.

First we note that the sequence of matrices we use in computing the greatest right invariant (and also left invariant) fuzzy quasi-order matrix may be infinite, and in such cases, the efficiency of reduc\-tions using such matrices becomes questionable.~In contrast, $k$-reductions by means of the $k$th members of sequences defined by \eqref{eq:seq.Q} and \eqref{eq:seq.P} can always be realized in a finite number of steps.~In addition, members of those arrays with smaller indices produce automata with fewer states than members with larger indices (see example below), which means that $k$-reductions generally yield fewer fuzzy finite automata than reductions that result in strictly equivalent fuzzy finite automata.

Regarding $k$-reductions by means of the $k$th members of sequences defined by \eqref{eq:seq.Q} and \eqref{eq:seq.hatQ}, by $Q_k\leqslant\widehat{Q}_k$ it follows that $\widehat{Q}_k$ generally yields better reduction than $Q_k$ (see Example \ref{ex:ex1}), for each~$k\in \mathbb{N}_0$. However, computing the matrix $\widehat{Q}_k$ can be significantly more difficult than computing the matrix $Q_k$, because it requires computing the vectors $\tau_u$, for $u\in X^*$, $|u|\leqslant k$, and their the number can be more than $m^k$, where $m$ is the number of input letters. 

Finally, as far as $k$-reductions by means of the $k$th members of sequences defined by \eqref{eq:seq.Q} and \eqref{eq:seq.P} are concerned, they are not comparable.~In some cases one of them will give better results, and in other cases the another one.~In Example \ref{ex:ex1} we have a case where $P_k$, for each $k\geqslant 1$, does not perform reduction, while $Q_k$ does, but the opposite would happen on the reverse automaton of the automaton considered in that example.~The same can be said for $k$-reductions by means of the $k$th~members of sequences~defined by \eqref{eq:seq.hatQ} and \eqref{eq:seq.hatP}.

In the following example, we assumed the structure of membership values to be the two-element Boolean algebra $\mathbb{B}$.~The reason why we decided so is that $\mathbb{B}$ is a subalgebra of every complete residua\-ted lattice, and every automaton over the two-element Boolean algebra $\mathbb{B}$ can also be considered an automaton over an arbitrary complete residuated lattice.

\smallskip

\begin{example}\label{ex:ex1}\rm
Let $\mathpzc{A}=(A,\sigma,\delta,\tau)$ be a fuzzy finite automaton over the two-element Boolean algebra $\mathbb{B}=\{0,1\}$ and an inpit alphabet $X=\{x,y\}$ given by
{\footnotesize\[
\sigma =\begin{bmatrix} 1 & 1 & 0 & 0 & 0 & 1 \end{bmatrix},\quad
\delta_x =\begin{bmatrix} 
1 & 0 & 0 & 0 & 0 & 0 \\
0 & 0 & 1 & 0 & 1 & 0 \\
1 & 0 & 1 & 1 & 0 & 0 \\
0 & 1 & 1 & 0 & 1 & 0 \\
0 & 0 & 0 & 1 & 0 & 0 \\
0 & 0 & 0 & 0 & 0 & 1 
\end{bmatrix},\quad
\delta_y =\begin{bmatrix} 
1 & 0 & 0 & 0 & 0 & 0 \\
1 & 1 & 0 & 0 & 0 & 0 \\
1 & 0 & 0 & 0 & 1 & 1 \\
0 & 0 & 0 & 1 & 0 & 1 \\
0 & 0 & 0 & 1 & 0 & 0 \\
1&  0 & 1 & 1 & 0 & 0 
\end{bmatrix},\quad
\tau =\begin{bmatrix} 1 \\ 1 \\ 0 \\ 1 \\ 0 \\ 1 \end{bmatrix}.
\]}%
Applying formula \eqref{eq:seq.Q} we get
{\footnotesize\begin{align*}
Q_0&=\begin{bmatrix} 
1 & 1 & 1 & 1 & 1 & 1  \\
1 & 1 & 1 & 1 & 1 & 1  \\
0 & 0 & 1 & 0 & 1 & 0  \\
1 & 1 & 1 & 1 & 1 & 1 \\
0 & 0 & 1 & 0 & 1 & 0  \\
1 & 1 & 1 & 1 & 1 & 1 
\end{bmatrix}, \quad 
Q_1=\begin{bmatrix} 
1 & 1 & 1 & 1 & 1 & 1  \\
0 & 1 & 0 & 0 & 0 & 0  \\
0 & 0 & 1 & 0 & 1 & 0  \\
1 & 1 & 1 & 1 & 1 & 1  \\
0 & 0 & 1 & 0 & 1 & 0  \\
1 & 1 & 1 & 1 & 1 & 1 
\end{bmatrix}, \quad 
Q_2=\begin{bmatrix} 
1 & 1 & 1 & 1 & 1 & 1  \\
0 & 1 & 0 & 0 & 0 & 0  \\
0 & 0 & 1 & 0 & 1 & 0  \\
0 & 1 & 0 & 1 & 0 & 0  \\
0 & 0 & 1 & 0 & 1 & 0  \\
1 & 1 & 1 & 1 & 1 & 1 
\end{bmatrix}, \quad \\ 
\end{align*}}%

{\footnotesize\begin{align*}
Q_3&=Q_4=\begin{bmatrix} 
1 & 1 & 1 & 1 & 1 & 1  \\
0 & 1 & 0 & 0 & 0 & 0  \\
0 & 0 & 1 & 0 & 1 & 0  \\
0 & 1 & 0 & 1 & 0 & 0  \\
0 & 0 & 0 & 0 & 1 & 0  \\
1 & 1 & 1 & 1 & 1 & 1
\end{bmatrix}, \quad 
\end{align*}}%
with $d(Q_0)=2$, $d(Q_1)=3$, $d(Q_2)=4$ and $d(Q_3)=5$, where ${Q}_3=Q_4$ is the greatest right invariant fuzzy quasi-order matrix, and applying \eqref{eq:seq.hatQ} we get
{\footnotesize\[
\widehat{Q}_0=\begin{bmatrix} 
1 & 1 & 1 & 1 & 1 & 1 \\
1 & 1 & 1 & 1 & 1 & 1 \\
0 & 0 & 1 & 0 & 1 & 0 \\
1 & 1 & 1 & 1 & 1 & 1 \\		
0 & 0 & 1 & 0 & 1 & 0 \\
1 & 1 & 1 & 1 & 1 & 1
\end{bmatrix}, \quad 
\widehat{Q}_1=\begin{bmatrix} 
1 & 1 & 1 & 1 & 1 & 1 \\ 
0 & 1 & 0 & 0 & 0 & 0 \\
0 & 0 & 1 & 0 & 1 & 0 \\
1 & 1 & 1 & 1 & 1 & 1 \\	
0 & 0 & 1 & 0 & 1 & 0 \\
1 & 1 & 1 & 1 & 1 & 1
\end{bmatrix}, 
\]}%
with $d(\widehat{Q}_0)=2$ and $d(\widehat{Q}_1)=3$, where $\widehat{Q}_1$ is the greatest weakly right invariant fuzzy quasi-order matrix.

We have that $\mathpzc{A}_{Q_k}$ is $k$-equivalent to $\mathpzc{A}$, for each $k\in \{0,1,2\}$, while $\mathpzc{A}_{Q_3}$ is strictly equivalent to $\mathpzc{A}$, and we also have that $|\mathpzc{A}_{Q_0}|<|\mathpzc{A}_{Q_1}|<|\mathpzc{A}_{Q_2}|<|\mathpzc{A}_{Q_3}|$.~On the other hand, $\mathpzc{A}_{\widehat{Q}_0}$ is $0$-equivalent to $\mathpzc{A}$, while  $\mathpzc{A}_{\widehat{Q}_1}$ is strictly equivalent to $\mathpzc{A}$.~We see that ${\widehat{Q}_1}$ provides better reduction than ${Q_3}$.

The sequence computed according to formula \eqref{eq:seq.P} is the following:
{\footnotesize\[
P_0=\begin{bmatrix}
1 & 1 & 0 & 0 & 0 & 1 \\
1 & 1 & 0 & 0 & 0 & 1 \\
1 & 1 & 1 & 1 & 1 & 1 \\
1 & 1 & 1 & 1 & 1 & 1 \\
1 & 1 & 1 & 1 & 1 & 1 \\
1 & 1 & 0 & 0 & 0 & 1 
\end{bmatrix}, \quad  
P_1=\begin{bmatrix}
1 & 0 & 0 & 0 & 0 & 0 \\
1 & 1 & 0 & 0 & 0 & 0 \\
1 & 0 & 1 & 0 & 0 & 0 \\
1 & 1 & 1 & 1 & 0 & 0 \\ 
1 & 0 & 1 & 0 & 1 & 1 \\
1 & 0 & 0 & 0 & 0 & 1 
\end{bmatrix}, \quad 
P_2=P_3=\begin{bmatrix}
1 & 0 & 0 & 0 & 0 & 0 \\
1 & 1 & 0 & 0 & 0 & 0 \\
1 & 0 & 1 & 0 & 0 & 0 \\
1 & 0 & 0 & 1 & 0 & 0 \\			
1 & 0 & 0 & 0 & 1 & 0 \\
1 & 0 & 0 & 0 & 0 & 1
\end{bmatrix}, \quad  
\]}%
with $d(P_0)=2$, $d(P_1)=6$ and $d(P_2)=6$, $P_2=P_3$ is the greatest left invariant fuzzy quasi-order matrix, while the sequence computed using \eqref{eq:seq.hatP} is the following:
{\footnotesize\[
\widehat{P}_0=\begin{bmatrix} 
1 & 1 & 0 & 0 & 0 & 1 \\
1 & 1 & 0 & 0 & 0 & 1 \\
1 & 1 & 1 & 1 & 1 & 1 \\
1 & 1 & 1 & 1 & 1 & 1 \\
1 & 1 & 1 & 1 & 1 & 1 \\
1 & 1 & 0 & 0 & 0 & 1
\end{bmatrix}, \quad 
\widehat{P}_1=\begin{bmatrix} 
1 & 0 & 0 & 0 & 0 & 0 \\
1 & 1 & 0 & 0 & 0 & 0 \\
1 & 0 & 1 & 0 & 0 & 0 \\
1 & 1 & 1 & 1 & 0 & 0 \\
1 & 0 & 1 & 0 & 1 & 1 \\
1 & 0 & 0 & 0 & 0 & 1 
\end{bmatrix}, \quad 
\widehat{P}_2=\widehat{P}_3=\begin{bmatrix} 
1 & 0 & 0 & 0 & 0 & 0 \\
1 & 1 & 0 & 0 & 0 & 0 \\
1 & 0 & 1 & 0 & 0 & 0 \\
1 & 0 & 0 & 1 & 0 & 0 \\
1 & 0 & 0 & 0 & 1 & 0 \\
1 & 0 & 0 & 0 & 0 & 1 
\end{bmatrix}, 
\]}%
with $d(\widehat{P}_0)=2$, $d(\widehat{P}_1)=6$ and $d(\widehat{P}_2)=6$, while $\widehat{P}_2$ is the greatest weakly left invariant fuzzy quasi-order matrix.~Although $\mathpzc{A}_{P_1}$ and $\mathpzc{A}_{\widehat{P}_1}$ are $1$-equivalent to $\mathpzc{A}$, and $\mathpzc{A}_{P_2}$ and $\mathpzc{A}_{\widehat{P}_2}$ are strictly equivalent to $\mathpzc{A}$, this is of no significance as there is no any reduction. 
\end{example}

\section{Complexity issues}

Let $n$ denote the number of states of a fuzzy finite automaton $\mathpzc{A}=(A,\sigma,\delta,\tau)$ and $m$ the number of letters in the input alphabet $X$, and let $c_{\lor}$, $c_{\land}$, $c_{\otimes}$ and $c_{\to}$ be respectively computational costs of the operations $\lor $, $\land$, $\otimes $ and $\to$ in the underlying complete residuated lattice $\mathbb{L}$.~If $\mathbb{L}$ is linearly ordered, we can assume that $c_{\lor}=c_{\land}=1$, and if $\mathbb{L}$ is the G\"odel structure, we can also assume that $c_{\otimes}=c_{\to}=1$.

First we consider the computational time of the procedure from Theorem \ref{th:ri.k-e} (or Theorem \ref{th:li.k-e}),~for a given $k\in \mathbb{N}$. It is clear that the time required to compute $Q_0$ is $O(n^2c_{\to})$.~When we have computed~$Q_s$, for some $s\in \mathbb{N}$, $s<k$, and we are computing $Q_{s+1}$ from $Q_s$, we have the following: 
\begin{itemize}\parskip=-2pt
\item[1)] For a fixed $x\in X$, the time required to compute the product $\delta_x\cdot Q_s$ is $O(n^3(c_{\otimes }+c_{\lor}))$, and when this product is computed, we need an additional time $O(n^3(c_{\to }+c_{\land}))$ to compute the residual $(\delta_x\cdot Q_s)/\delta_x$. Thus, the total time required to compute $(\delta_x\cdot Q_s)/\delta_x$ is $O(n^3(c_{\to }+c_{\land}+c_{\otimes }+c_{\lor}))$.
\item[2)] Now, to compute $(\delta_x\cdot Q_s)/\delta_x$ for all $x\in X$ we need time $O(mn^3(c_{\to }+c_{\land}+c_{\otimes }+c_{\lor}))$.
\item[3)] Next, when all matrices from 2) have been computed, to compute all infima in $Q_s\land \bigwedge_{x\in X}(\delta_x\cdot Q_s)/\delta_x$ and obtain $Q_{s+1}$ we need time $O(n^2mc_{\land})$.
\item[4)] Finally, the total time required to compute $Q_{s+1}$ from $Q_s$ is $O(mn^3(c_{\to }+c_{\land}+c_{\otimes }+c_{\lor}))$.
\end{itemize}
Therefore, the time required to compute all matrices $Q_0$, $Q_1$, \dots, $Q_k$, i.e., the total computational time of the procedure from Theorem \ref{th:ri.k-e}, amounts $O(kmn^3(c_{\to }+c_{\land}+c_{\otimes }+c_{\lor}))$.~This is also the total computational time of the procedure from Theorem \ref{th:li.k-e}.~That computational time can be even better in cases where $Q_s=Q_{s+1}$, for some $s<k$, because then we do not have to compute all the matrices between $Q_{s+1}$ and $Q_k$ which are all equal to each other.~However, to achieve that better computational time, after computing the matrix $Q_{s+1}$ we need to check whether $Q_{s+1}=Q_{s}$.~The time required for such a check is $O(n^2)$, and for all such checks it is at most $O(kn^2)$, which obviously does not affect the total computational time order of our procedure.

Next we consider the computational time of the procedure from Theorem \ref{th:wri.k-e} (or Theorem \ref{th:wli.k-e}),~for a given $k\in \mathbb{N}$.~To compute $\widehat{Q}_k$ we have to compute the family of vectors $\{\tau_u\}_{|u|\leqslant k}$.~That family forms a perfect $m$-ary tree with the root corresponding to the vector $\tau$, and computing the members of the family is reduced to filling that tree level by level, starting from the root.~Consequently, the number of members of that family is at most $O(m^k)$.~When a vector $\tau_u$, for some $u\in X^*$, $|u|<k$, is computed, then for each $x\in X$ we compute $\tau_{xu}$ according to the formula $\tau_{xu}=\delta_x\cdot \tau_u$.~Therefore, the time needed to compute this product, i.e., the time nedded to compute each member of the considered family, is $O(n^2(c_{\otimes}+c_{\lor}))$.~Moreover, when $\tau_u$, for some $u\in X^+$, $|u|\leqslant k$, is computed, the residual $\tau_u/\tau_u$ can be computed in time $O(n^2c_{\to})$, so the time required to compute a single $\tau_u$ and the residual $\tau_u/\tau_u$ is $O(n^2(c_{\otimes}+c_{\lor}+c_{\to }))$, and the total time required to compute the whole family $\{\tau_u\}_{|u|\leqslant k}$ and all residuals $\tau_u/\tau_u$, for $u\in X^*$, $|u|\leqslant k$, amounts $O(m^kn^2(c_{\otimes}+c_{\lor}+c_{\to }))$.~Finally, to compute $\widehat{Q}_k$ we have to apply the operation $\land$ between matrices $\tau_u/\tau_u$, $u\in X^*$, $|u|\leqslant k$, at most $O(m^k)$ times, and since the computational time for a single application of the operation $\land $ is $O(n^2c_{\land})$, the total computational~time for all such operations is $O(m^kn^2c_{\land})$.~Hence, the total time required to compute $\widehat{Q}_k$, i.e., the total computational~time of the procedure from Theorem \ref{th:wri.k-e} is $O(m^kn^2(c_{\otimes}+c_{\lor}+c_{\to }+c_{\land}))$.~~This is also the total computational time of the procedure from Theorem \ref{th:wli.k-e}.

Applying the technique of~\cite{NMS.23b}, the procedures from Theorems \ref{th:ri.k-e}--\ref{th:wli.k-e} can be improved so that the factor $n^2$ in their complexity estimate is reduced to the sum of $n$ and the number of non-zero fuzzy transitions of the input fuzzy automaton~$\mathpzc{A}$.

\bibliographystyle{eptcs}
\bibliography{approxbib}

\end{document}